\documentclass[journal,10pt]{IEEEtran}

\usepackage{pgfplots} 
\usepackage{comment}
\usepackage{amsfonts}
\usepackage{textcomp}
\def\BibTeX{{\rm B\kern-.05em{\sc i\kern-.025em b}\kern-.08em
		T\kern-.1667em\lower.7ex\hbox{E}\kern-.125emX}}
\usepackage{amsmath,amssymb,multirow}
\usepackage{setspace}
\usepackage{listings}
\usepackage{cite}
\usepackage{acronym}
\usepackage{mathtools}
\usepackage{subfigure}
\usepackage{algorithm,algorithmic}
\usepackage{bm}

\pgfplotsset{
	tick label style={font=\small},
	label style={font=\small},
	legend style={font=\small}
}

\newcommand\abs[1]{\big|#1\big|}

\newcommand{\Tcp}{T_{\mathrm{cp}}}
\newcommand{\Ts}{T_{\mathrm{s}}}
\newcommand{\dtx}{d_{\mathrm{tx}}}
\newcommand{\drx}{d_{\mathrm{rx}}}
\newcommand{\dbis}{d_{\mathrm{bis}}}
\newcommand{\vbis}{v_{\mathrm{bis}}}
\newcommand{\Htilde}{\widetilde{H}}

\newcommand{\Ztilde}{\widetilde{Z}}
\newcommand{\Ytilde}{\widetilde{Y}}

\newcommand{\mub}{\boldsymbol{\mu}}
\newcommand{\etab}{\boldsymbol{\eta}}
\newcommand{\hermit}{\mathsf{H}}
\newcommand{\Jb}{\boldsymbol{\mathrm{J}}}

\newcommand{\aP}{\abs{\mathcal{P}}}

\newcommand{\tcb}[1]{\textcolor{black}{#1}}

\newcommand{\WRFb}{\boldsymbol{W}_{\text{RF}}}
\newcommand{\WSb}{\boldsymbol{W}_{\text{S}}}
\newcommand{\FRFb}{\boldsymbol{F}_{\text{RF}}}

\newcommand{\Ab}{\boldsymbol{A}}
\newcommand{\Bb}{\boldsymbol{B}}
\newcommand{\Cb}{\boldsymbol{C}}
\newcommand{\Db}{\boldsymbol{D}}

\newcommand{\pb}{\boldsymbol{p}}

\makeatletter \renewcommand\d[1]{\ensuremath{%
		\;\mathrm{d}#1\@ifnextchar\d{\!}{}}}
\makeatother

\makeatletter
\newcommand*\rel@kern[1]{\kern#1\dimexpr\macc@kerna}
\newcommand*\widebar[1]{%
	\begingroup
	\def\mathaccent##1##2{%
		\rel@kern{0.8}%
		\overline{\rel@kern{-0.8}\macc@nucleus\rel@kern{0.2}}%
		\rel@kern{-0.2}%
	}%
	\macc@depth\@ne
	\let\math@bgroup\@empty \let\math@egroup\macc@set@skewchar
	\mathsurround\z@ \frozen@everymath{\mathgroup\macc@group\relax}%
	\macc@set@skewchar\relax
	\let\mathaccentV\macc@nested@a
	\macc@nested@a\relax111{#1}%
	\endgroup
}
\makeatother

\usepackage{amsthm}

\theoremstyle{remark}

\newtheoremstyle{mytheoremstyle} 
{\topsep}                    
{\topsep}                    
{\upshape}                   
{.5em}                           
{\itshape}                   
{.}                          
{.5em}                       
{}  

\theoremstyle{mytheoremstyle}

\newtheoremstyle{iremark}
{\topsep}   
{\topsep}   
{\upshape}  
{0.2in}       
{\itshape}  
{.}         
{5pt plus 1pt minus 1pt} 
{\thmname{#1}\thmnumber{ \itshape#2}\thmnote{ (#3)}} 

\theoremstyle{plain}
\newtheorem{prop}{Proposition}

\newtheorem{rem}{Remark}

\allowdisplaybreaks 
\begin{document}
	\title{Impact of the Pilot Design for OFDM Based Bi-static Integrated Sensing and Communication System}
	\thispagestyle{plain}

	\author{Cuneyd Ozturk  \emph{Member, IEEE}, 
		and Cagri Goken \emph{Member, IEEE}. 
		\thanks{Authors are with Aselsan Inc., Ankara, 06800, Turkey (E-mail: cuneydozturk@aselsan.com, cgoken@aselsan.com)
			
		{More concise version of this manuscript has been submitted to IEEE Wireless Communication Letters.}
	}}
	\maketitle
	\begin{abstract}
		A bistatic millimeter-wave (mmWave) ISAC system utilizing OFDM signaling is considered. 
		For a single-target scenario, closed-form expressions for the Cramer-Rao bounds (CRBs) of range and velocity estimation are derived for a given pilot pattern. The analysis shows that when the target's range and velocity remain within the maximum unambiguous limits, allocating pilot symbols more frequently in time improves position estimation, while increasing their density in frequency enhances velocity estimation. Numerical results further validate that the least squares (LS) channel estimation approach closely follows CRB predictions, particularly in the high-SNR regime.
	\end{abstract}
	

	\section{Introduction}\label{sec:Intro}

	Recent studies highlight the significant potential of OFDM-based ISAC systems \cite{Gonzalez2024ISAC, Keskin2021ICIFoe, Keskin2021LimitedFeedForward, keskin2024fundamentaltradeoffsmonostaticisac, Swindlehurst2024Joint3D, Swindlehurst2024MIMOISAC, bacchielli2024bistaticsensingthzfrequencies, Pucci2022_Bistatic, brunner2024bistaticofdmbasedisacovertheair, Molisch2024ISACBistaticRadar}. Several of these works \cite{Gonzalez2024ISAC, Keskin2021ICIFoe, Keskin2021LimitedFeedForward, keskin2024fundamentaltradeoffsmonostaticisac, Swindlehurst2024Joint3D, Swindlehurst2024MIMOISAC} focus on the monostatic configuration, where the transmitter and receiver are co-located, which requires full-duplex capability. In contrast, others \cite{bacchielli2024bistaticsensingthzfrequencies, Pucci2022_Bistatic, brunner2024bistaticofdmbasedisacovertheair, Molisch2024ISACBistaticRadar} explore the bi-static scenario, where the transmitter and receiver are separate, thereby eliminating self-interference concerns.

	\tcb{Typically, the transmitter designates part of its resources to pilot symbols in communication systems for channel estimation. For ISAC systems, pilot symbols can further be utilized for sensing purposes. The sensing receiver leverages these  symbols to estimate the channel between the target and receiver in the monostatic scenario, and the combined channels from the transmitter to the target and from the target to the receiver in the bistatic scenario. The remaining resources, apart from the pilot symbols, are the data symbols and assigned for communication, creating a fundamental trade-off between sensing and communication due to the allocation of pilot symbols.}

	\tcb{In the bi-static configuration, the sensing receiver can utilize the entire OFDM frame with decoding the data symbols before estimating the associated channel coefficients, which adds complexity and poses an additional challenge, particularly when the communication and radar receivers are separate. Alternatively, relying solely on pilot symbols reduces processing gain, it simplifies the receiver design.}
	%

	Since an equally spaced pilot signal limits the maximum unambiguous range and velocity, an alternative pilot design scheme is proposed using coprime and periodic stepping values for pilot indices in \cite{mei2024CoprimePeriodicPilot}. It is demonstrated that the proposed pilot design does not reduce maximum unambiguous range and velocity. In \cite{lyu2024referencesignalbasedwaveformdesign}, pilot design is studied for the monostatic configuration, where the user serves as both a downlink communication device and a sensing target. Cramer-Rao Bounds (CRBs) for velocity and range estimation of the target are derived, and the objective is to minimize their weighted sum while satisfying a communication rate constraint.
	
	In this work, we investigate the impact of the pilot symbol design on a bistatic milimeter wave (mmWave) ISAC utilizing OFDM. \tcb{We assume that bi-static radar processing is conducted only by using the pilot symbols, instead of the whole OFDM frame.} We explore the fundamental limits of parameter estimation by deriving the CRBs for range and velocity estimation under a given pilot pattern design. By analyzing the influence of pilot placement and system parameters on estimation accuracy, our findings contribute to the optimal pilot design of ISAC waveforms for improved sensing performance. Furthermore, numerical results confirm that when  2D-FFT approach is applied to least square (LS) channel estimates, the range and velocity estimation errors closely follow CRB trends, particularly in the high-SNR regime.
	\section{System Model}
	
	\subsection{Geometry}
	\begin{figure}
		\centering
		\includegraphics[width = 0.8\linewidth]{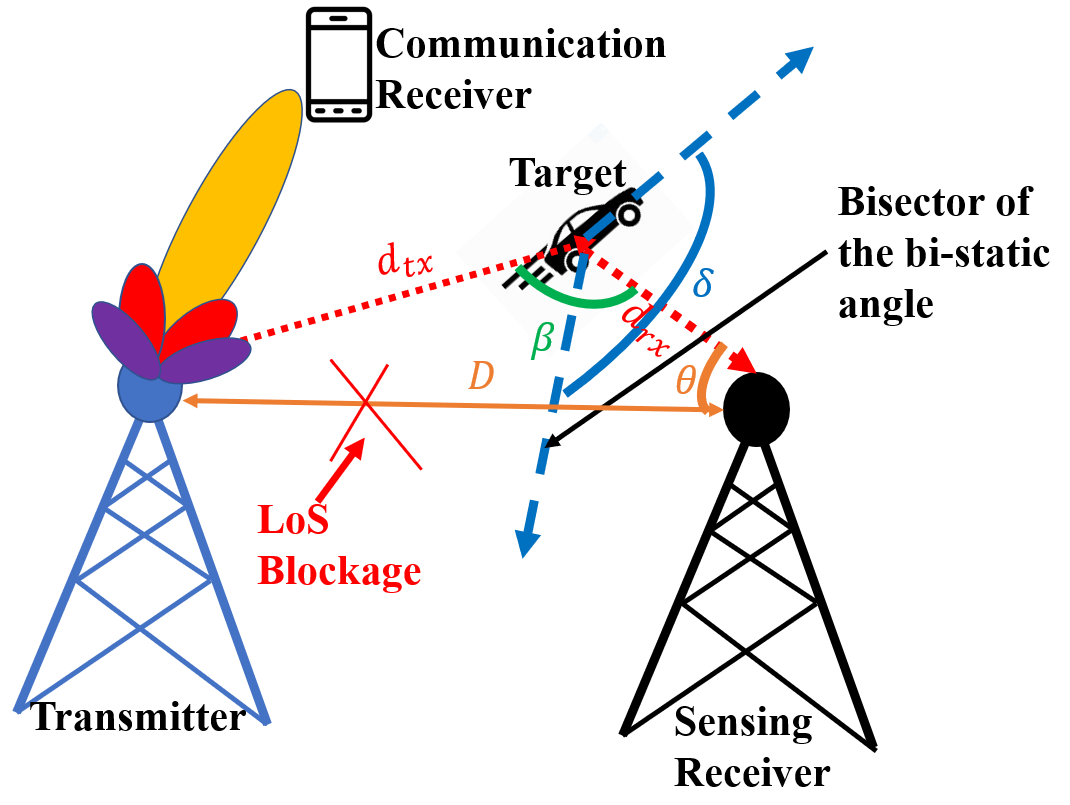}
		\caption{Considered bi-static ISAC system geometry illustrating the distances $\dtx, \drx, D$, the bi-static angle $\beta$, the angle between the target direction and the bisector of the bi-static angle $\delta$\tcb{, and AoA $\theta$}.}
		\label{fig:geometry}
	\end{figure}
	We examine a bi-static ISAC system that employs an OFDM waveform and operates at millimeter-wave frequencies. The system setup includes a transmitter, a single target, a receiver, and user equipment (UE). In particular,  the receiver aims to estimate the position of the target based on the signals reflected by the target.

	As illustrated in Fig.~\ref{fig:geometry}, we define $\dtx$, $\drx$, and $D$ as the distance between the transmitter and the target, the distance between the target and the receiver, and the distance between the transmitter and the receiver, respectively. The transmitter and receiver are stationary, and the receiver knows the distance to the transmitter.   The angle at the target subtended by the transmitter and the receiver, denoted as $\beta$ (bi-static angle), is given by $\beta = \cos^{-1}\left((\dtx^2 + \drx^2 - D^2)/(2\dtx \drx)\right)$.
	
	In addition, we define $\delta$ as the angle between the direction of the target's velocity vector and the bi-static bisector. Furthermore, $\theta$ is defined as the angle-of-arrival \tcb{(AoA)} at the receiver.  By following the notation given by \cite{Pucci2022_Bistatic}, we define $\dbis \triangleq \dtx + \drx$ as the bi-static range. The relationship between $\dbis$ and $\drx$ can be given as $\drx = (\dbis^2-D^2)/\left(2(\dbis-D\cos\theta)\right)$
	Therefore, if $D$ and $\theta$ are known at the receiver, $\drx$ can be easily estimated once $\dbis$ is estimated. 
	
	\tcb{The sensing receiver can scan the environment using analog beamforming to estimate the AoA by employing a constant false alarm rate (CFAR) detector. In this study, we assume that the AoA is obtained at the sensing receiver. Therefore, the sensing receiver aligns its beam center accordingly. Consequently, throughout this manuscript, the receiver focuses on estimating the target's range and velocity. Additionally, the transmitter directs its beam toward the communication receiver.}

	\subsection{Signal Model}
	$N$ and $M$ denote the number of OFDM subcarriers and the number of OFDM symbols \tcb{in a communication frame}, respectively. $\Delta f$ denotes the subcarrier spacing, and $\Ts$ is the OFDM symbol duration including the cyclic prefix (CP). In particular, $\Ts = T + \Tcp$, where $T = 1/\Delta f$ and $\Tcp$ is the duration of the CP. 

	\subsubsection{Modulation Symbols}
	We assume that a single data stream is transmitted and we define $X_{n,m}$ as the modulation symbol corresponding to the $n$-th subcarrier and $m$-th OFDM symbol. Throughout this manuscript, for any $(n,m)$ we assume $X_{n,m}$ is QPSK symbols i.e., $\abs{X_{n,m}} = 1$. We assume some of these modulation symbols are pilot symbols and known at the receiver. In particular, we denote the set of pilot symbols as $\mathcal{P}  \triangleq\{(n,m)\mid  X_{n,m}~\text{is a pilot symbol}\}$. We define $0\le \rho \le 1$ as the ratio of the sources allocated to the pilot symbols, i.e., $\rho = \aP/(NM),$ where $\abs{\cdot}$ denotes the number of elements of its argument. 
	
	\subsubsection{Sensing Channel}
	The transmitter and receiver each have a single radio frequency (RF) \tcb{chain}, hence analog beamforming is utilized at both the transmitter and receiver\footnote{Fully digital beamforming requires a separate RF chain for each antenna element, making it impractical at mmWave frequencies due to high costs and significant power consumption \cite{Sohrabi2016BF}.}. \tcb{}
	
	Let $N_T$ and $N_R$ represent the number of antennas at the transmitter and receiver, respectively. We also assume that the line-of-sight (LoS) path between the transmitter and the receiver is blocked, and only the reflected path from the target is considered{\footnote{\tcb{If a LoS path is present, the sensing receiver can estimate the target's range by synchronizing to the LoS path, as considered in \cite{Natajara2024BistaticRadar}. However, if the LoS path is sufficiently strong, the dynamic range may constrain the sensing receiver’s capability.}}. By following the model described in \cite{Gonzalez2024ISAC} and  \cite{Sohrabi2016BF}, the overall channel from the transmitter to the receiver for the $n$-th subcarrier and the $m$-th OFDM symbol, denoted as $H_{n,m}$, can be expressed as
		\begin{align}
		H_{n,m} = \kappa e^{-j2\pi \tau n \Delta f } e^{j 2\pi f_D m \Ts} \WRFb^\hermit \WSb \FRFb ,
		\end{align}
		where $\WRFb\in\mathbb{C}^{N_R \times 1}$ is the analog precoder at the receiver, $\WSb\in \mathbb{C}^{N_R\times N_T}$ denotes the multiplication of the transmit and the receive steering vectors,  $ \FRFb\in\mathbb{C}^{N_T\times 1}$ is the analog precoder at the transmitter, $\kappa$ is the channel gain, $f_D$ and $\tau$ denote the Doppler shift and the propagation delay, respectively. The overall channel gain is defined as $\alpha\triangleq \kappa  \WRFb^\hermit \WSb \FRFb \in\mathbb{C}$.
		
		By following the notation given by \cite{Pucci2022_Bistatic}, we define $\vbis \triangleq v \cos\delta$ as the bi-static velocity, where $v$ is the amplitude of the target velocity vector and $\delta$ is the angle between the target's velocity vector and the bisector of the bi-static angle as shown in Fig.~\ref{fig:geometry}. Then, the Doppler shift is expressed as  $	f_D =  (2\vbis\cos\left(\beta/2\right))/\lambda $\cite{willis2005bistatic}.
		In addition, the propagation delay is expressed as $\tau = \dbis/c$,
		where $c = 3\times 10^8$ is the speed of light. We assume $\tau\le \Tcp$ to avoid inter-symbol interference~\tcb{(ISI)}. This implies that the maximum detectable range $d_{\text{max}}$ is simply equal to $c\times \Tcp$.
		
		The received symbol for the $n$-th subcarrier and the $m$-th OFDM symbol after OFDM demodulation and CP removal are given by
		\begin{align}
		Y_{n,m} = H_{n,m} X_{n,m} + Z_{n,m},
		\end{align}
		where $Z_{n,m}\sim\mathbb{C}\mathcal{N}(0, \sigma^2)$.

		\subsubsection{Communication Channel}
		We define $\Htilde_{n,m}$ as the communication channel \tcb{coefficient} between the transmitter and the target for the $n$-th subcarrier and the $m$-th OFDM symbol. \tcb{In particular, it is assumed that $\mathbb{E}\left\{\abs{\Htilde_{n,m}}^2  \right\} = \sigma^2_{\tilde{H}}$
			for any $n,m$.}	
		
		For the communication channel, the received symbol for the $n$-th subcarrier and the $m$-th OFDM symbol after OFDM demodulation and CP removal are given by
		\begin{align}
		\Ytilde_{n,m} = \Htilde_{n,m} X_{n,m} + \Ztilde_{n,m},
		\end{align}
		where $\Ztilde_{n,m}\sim\mathbb{C}\mathcal{N}(0,\tcb{\tilde{N}_0})$ \tcb{for any $n,m$.}
		
		
		%
		%
		%

		\subsection{Performance Metrics}
		
		\subsubsection{\tcb{Sensing Metric}}
		For the sensing channel, we use the CRB related to the range and velocity estimation of the target as the CRB becomes tight to the mean-squared estimator of the ML estimator. In particular, we denote $\etab \triangleq [\alpha_R,~ \alpha_I,~f_D,~ \tau]$ as the unknown parameters to be estimated, where $\alpha_R \triangleq \Re\{\alpha\}$ and $\alpha_I \triangleq \Im\{\alpha\}$. Then, the Fisher information matrix (FIM) for the parameter vector $\etab$ is denoted as $\Jb(\etab) \in\mathbb{C}^{4\times 4}$. Then, the CRBs related to the range and velocity estimations are given by 
		\begin{align}\label{eq:CRB_ran_vel}
		\text{CRB}_{\text{ran}} &= c^2[\Jb^{-1}(\etab)]_{4,4},~\text{CRB}_{\text{vel}} =  \frac{\lambda^2[\Jb^{-1}(\etab)]_{3,3}}{4 \cos^2 \left(\beta/2\right)}.
		\end{align}
		
		\subsubsection{\tcb{Communication Metric}}
		For the communication channel we use an upper bound for the Shannon (ergodic) capacity 
		as the performance metric. In particular, 
		\tcb{by averaging over the realizations of the channel coefficients \cite{schafhuber2004wireless} and employing the Jensen's inequality}, we can write
		
		\tcb{
			\begin{align}
			R = & \frac{1}{M\Ts}\sum_{(n,m)\notin \mathcal{P}} \mathbb{E} \left\{  \log_2 \left(1 + \frac{\abs{\Htilde_{n,m}X_{n,m}}^2}{\tilde{N}_0} \right)\right\} \nonumber\\
			&\leq \underbrace{\frac{N(1-\rho)}{T_s}\log_2\left(1 + \frac{\sigma^2_{\tilde{H}}}{\tilde{N}_0}\right)}_{\triangleq R_u}~[\text{bits/sec}] \label{eq:rate_upper_bd}.
			\end{align}
		} One should note that while calculating $R_u$, channel estimation error for the communication channel is not taken into account. 
		
		Our objective is to examine how the pilot design affects the trade-off between sensing and communication networks, using the CRBs given by \eqref{eq:CRB_ran_vel} and the communication rate upper bound \tcb{[bits/sec]} given by \eqref{eq:rate_upper_bd}. In particular, our design parameter is selection of $\mathcal{P}$.

		\section{Theoretical Bounds}

		In this section, the CRBs related to the range and velocity estimations will be computed.  We define $\mub(\etab) \triangleq [\mu_{n,m}(\etab)]_{(n,m)\in\mathcal{P}}$ where $\mu_{n,m}(\etab)\triangleq H_{n,m} X_{n,m}$ for $(n,m)\in\mathcal{P}$.  The
		$(k,\ell)$-th entry of $\Jb(\etab)$ can be expressed as \cite{kay1993statistical}
		\begin{align}
		[\Jb(\etab)]_{k,\ell} = \frac{2}{\sigma^2} \Re \left\{\left(\frac{\partial \mub(\etab)}{\partial \eta_k}\right)^{\hermit} \left(\frac{\partial \mub(\etab)}{\partial \eta_\ell}\right)\right\} \in \mathbb{C}^{4\times 4},
		\end{align}
		where $\eta_k$ denotes the $k$-th element of $\etab$. The equivalent Fisher information matrix (EFIM) \cite{Shen2010fundamentallimits} for the estimation of the range and the velocity is denoted as $\Jb_{e}(\etab)\in\mathbb{C}^{2\times 2}$. \tcb{In particular, the following equation is satisfied $[\Jb_e^{-1}(\etab)]_{k,k}  = [\Jb^{-1}(\etab)]_{k+2,k+2}$for $k\in\{1, 2\}$.}
		
		\begin{prop}\label{prop:EFIM}
			$\Jb_{e}(\etab)$ can be expressed as
			\begin{align}\label{eq:EFIM}
			\Jb_{e}(\etab) = \frac{8\pi^2 \abs{\alpha}^2}{\sigma^2} 
			\begin{bmatrix}
			\Ts^2 Q_{M^2}      & -\Ts \Delta f Q_{NM} \\
			-\Ts \Delta f Q_{NM}      & \Delta f^2 Q_{N^2}  \\
			\end{bmatrix}
			\end{align}
			where
			\begin{align}
			Q_{M^2} &\triangleq \sum_{(n,m)\in\mathcal{P}} m^2- \frac{\left(\sum_{(n,m)\in\mathcal{P}} m\right)^2}{\aP},\\
			Q_{NM} &\triangleq \sum_{(n,m)\in\mathcal{P}} nm- \frac{\left(\sum_{(n,m)\in\mathcal{P}} n\right)\left(\sum_{(n,m)\in\mathcal{P}} m\right)}{\aP},\\
			Q_{N^2} &\triangleq \sum_{(n,m)\in\mathcal{P}} n^2- \frac{\left(\sum_{(n,m)\in\mathcal{P}} n\right)^2}{\aP}.
			\end{align}
		\end{prop}
		\begin{proof}
			One can refer Appendix~\ref{sec:FisherInformationMat} for detailed computations of $\Jb(\etab)$.  By using the Woodbury identity, the EFIM can be computed as
			\begin{align}
				\Jb_{e}(\etab) = \tcb{\frac{2}{\sigma^2}}(\Db-\Cb\Ab^{-1}\Bb),
			\end{align}
			where $\Ab, \Bb, \Cb$ and $\Db$ are defined in \eqref{eq:A}-\eqref{eq:D}. After some algebraic manipulations, \eqref{eq:EFIM} can be obtained.
		\end{proof}

		As a consequence of Proposition~\ref{prop:EFIM}, CRBs related to the range and velocity are calculated as
		\begin{align}
		\text{CRB}_{\text{ran}} &= \frac{Q_{M^2}}{Q_{N^2}Q_{M^2}-(Q_{NM})^2} \frac{\sigma^2 c^2}{8\pi^2 \Delta f^2 \abs{\alpha}^2} \label{eq:CRB_ran},\\
		\text{CRB}_{\text{vel}} &= \frac{Q_{N^2}}{Q_{N^2}Q_{M^2}-(Q_{NM})^2} \frac{\sigma^2 \lambda^2}{32\pi^2 \abs{\alpha}^2 \Ts  \cos^2\left(\beta/2\right)}\label{eq:CRB_vel}.
		\end{align}
		\vspace{-0.8cm}
		\subsection{Special Case: Periodic Pilot Pattern}
		Next, we present CRB values when pilot symbols are periodically distributed across time and frequency as widely used in communication systems. In particular, we define $\mathcal{P}$ as a periodic pilot pattern if $\mathcal{P}$ can be expressed as
		\begin{align}\label{eq:pilot_periodic_pattern}
		\mathcal{P} = \left\{(n, m)\mid \frac{n}{n_p} \in\mathbb{N}~\text{and}~ \frac{m}{m_p} \in\mathbb{N}\right\}.
		\end{align}
		for some $n_p, m_p\in\mathbb{N}$. The number of pilot symbols is given by $\aP = (K+1)(L+1)$, where $K$ and $L$ are defined as 	$K \triangleq \lfloor(N-1)/n_p\rfloor$ and $L \triangleq \lfloor(M-1)/m_p\rfloor$. For example, when $n_p = 3$ and $m_p =2$, an example of periodic pilot pattern is provided in Fig.~\ref{fig:pilot_ex}.
		
		\begin{rem}\label{rem:max_detect}
			For the periodic pilot design, the maximum detectable range and velocity will be simply $c/(n_p \Delta f)$ and $c/(2f_c m_p\Ts  \cos(\beta/2))$, respectively. In other words, as $n_p$ or $m_p$ increases, the system compromises its maximum unambigious range or velocity.
		\end{rem}
		\begin{figure}
			\centering
			\includegraphics[width=0.9\linewidth]{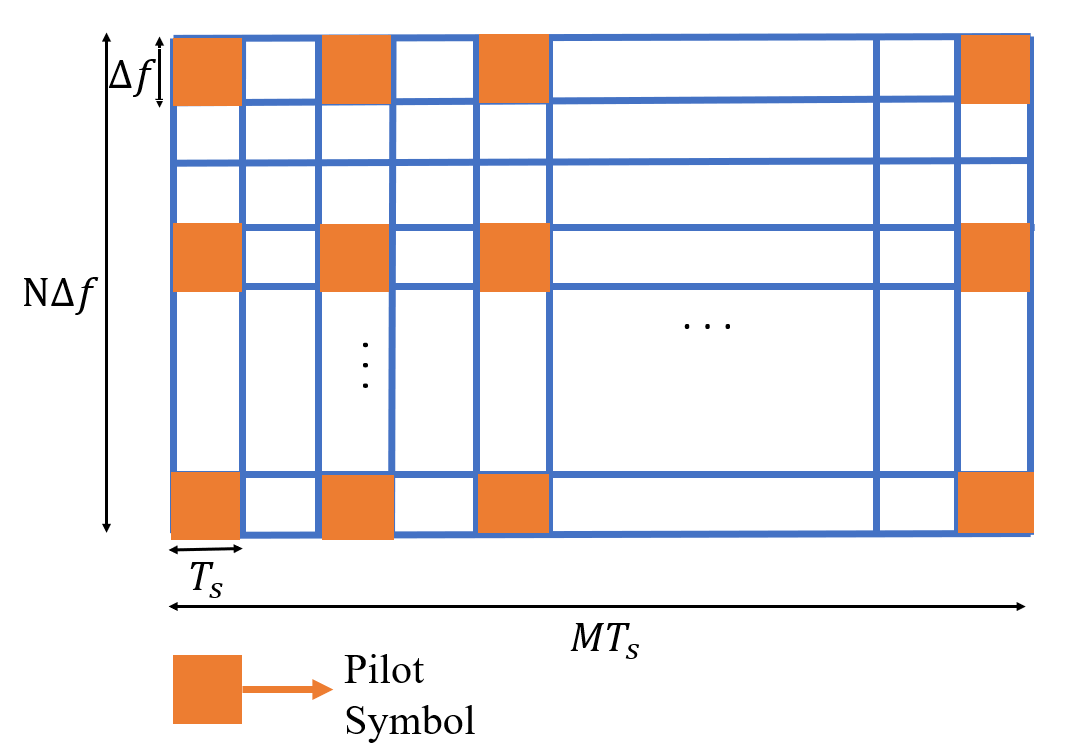}
			\caption{An example of pilot pattern over time and frequency for $n_p = 3$ and $m_p = 2$.}
			\label{fig:pilot_ex}
			\vspace{-0.5cm}
		\end{figure}

		\begin{prop}\label{prop:CRB_periodic}
			For a periodic pilot pattern as defined in \eqref{eq:pilot_periodic_pattern}, $\text{CRB}_{\text{ran}}$ and $ \text{CRB}_{\text{vel}}$ are given by
			\begin{align}
			\text{CRB}_{\text{ran}} =  \frac{12}{K(K+2)\aP n_p^2} \frac{\sigma^2 c^2}{8\pi^2 \abs{\alpha}^2\Delta f^2}
			\end{align}
			and
			\begin{align}\label{eq:CRB_vel_periodic}
			\text{CRB}_{\text{vel}} =  \frac{12}{L(L+2)\aP m_p^2} \frac{\sigma^2 \lambda^2}{32\pi^2 \abs{\alpha}^2 \Ts^2  \cos^2\left(\beta/2\right)}.
			\end{align}
		\end{prop}
		\begin{proof}
			Combining \eqref{eq:CRB_ran} and \eqref{eq:CRB_vel} with the results (\eqref{eq:Q_NM_periodic}, \eqref{eq:Q_N2_periodic},
			\eqref{eq:Q_M2_periodic})  from Appendix~\ref{sec:FIM_pilot_periodic}, we reach the desired expressions.
		\end{proof}
		\begin{rem}\label{rem:CRB}
			By ignoring the constant terms, $\text{CRB}_{\text{ran}}$ decays according to 
			\begin{align}
			\frac{1}{K (K+2)\aP n_p^2} &\approx\frac{1}{\aP (N-1) \left(N-1 + \frac{2(N-1)}{K}\right)} \\
			&\propto \frac{1}{\rho \left(1 + \frac{2}{K}\right)}
			\end{align}
			Similarly,  $\text{CRB}_{\text{vel}}$ decays according to  $1/(\rho \left(1 + \frac{2}{L}\right))$.
			\tcb{Consequently, for a given $\rho$, equivalently, for a given communication rate \eqref{eq:rate_upper_bd}, to have a better range estimate $K$ needs to be minimized whereas $L$ needs to be minimized to have a better velocity estimate.}
			

			
			
		\end{rem}

		\section{Simulation Results}
		In this section, we present numerical experiments to illustrate the impact of the pilot design, pilot overhead ratio $(\rho)$ and SNR [dB] on the sensing performance, where we define $\text{SNR} \triangleq \abs{\alpha}^2/\sigma^2$. 
		The transmitter and the receiver are located at $[-40, 0]~$[m] and $[0, 40]~$[m], respectively. Both the transmitter and the receiver are equipped with a single antenna. The target location is denoted at $\pb_{\mathrm{t}} = [x_{\mathrm{t}},~ y_{\mathrm{t}}]$ where $x_{\mathrm{t}}\sim\mathcal{U}[80, 100]~$[m] and $y_{\mathrm{t}}\sim\mathcal{U}[-100, -80]~$[m]. Target's velocity, $v$ is modeled as $v\sim\mathcal{U}[-30, 30]~$[m/s]. In addition, $\delta$, the angle between target's velocity vector and the bisector of the bi-static angle is modeled as $\delta\sim\mathcal{U}[-5^\circ,~5^\circ]$. \tcb{The carrier frequency $f_c$ is set to $30~$GHz. OFDM related parameters are taken as $\Delta f = 200~$kHz, $N = 70$, $M = 50$.} $\Tcp$ is taken as $1~[\mu$s], i.e., the maximum detectable range \tcb{to avoid ISI} is equal to $3\times 10^{8} \times 10^{-6} = 300~$[m]. To support this distance, $n_p\le 5$ should be satisfied. Similarly, from Remark~\ref{rem:max_detect}, to be able to detect target's velocity $\in[-30, 30]~$[m/s], $m_p\leq 27$ should be satisfied.

		\tcb{Range and velocity estimations are performed by using 2D-FFT-based periodogram method as described in \cite{braun2014ofdm}. After finding the peak of the periodogram, quadratic interpolation is employed \cite{braun2014ofdm}. FFT sizes over both subcarriers and symbols are taken as $4096$.}

		In Figs.~\ref{fig: RMSE_range_vs_SNR}-\ref{fig: RMSE_vel_vs_SNR}, the RMSEs for $\dbis$ and $\vbis$ are plotted against SNR for four different values of $\rho\in\{0.02, 0.1, 0.5, 1\}$ are plotted, along with their corresponding theoretical bounds\footnote{\tcb{In Figs.~\ref{fig: RMSE_range_vs_SNR}, \ref{fig: RMSE_vel_vs_SNR}, square roots of $\text{CRB}_{\text{ran}}$ and 	$\text{CRB}_{\text{vel}}$ are plotted.}}.}		\tcb{Similarly, when $\sigma^2_{\tilde{H}}/\tilde{N}_0$ is set to 5 dB, for the considered $\rho$ values,  upper bounds on the communication rate [Mbps] are computed as $\{23.523, 21.602, 12.001, 0\}~$Mbps, respectively.} $\rho = 0.02, 0.1, 0.5, 1$ is obtained when $(n_p, m_p) = (10, 5), (2, 5), (2, 1)$ and $(1, 1)$, respectively. Since the bi-static angle $\beta$ appears in \eqref{eq:CRB_vel}, $\text{CRB}_{\text{vel}}$ depends on the realization of the target's position. Therefore, we plot the expected CRB over different target position realizations, denoted as $\text{ECRB}_{\text{vel}}$.

	When $\rho = 0.02$, meaning $n_p = 10$ and $m_p = 5$, reliable range estimation is not achievable, as observed from Fig.~\ref{fig: RMSE_range_vs_SNR}. In addition, as we observe from  Fig.~\ref{fig: RMSE_vel_vs_SNR}, LS algorithm does not follow the CRB trend even though $m_p\le 27$. This is because velocity estimation relies on the bi-static angle estimate $\widehat{\beta}$, which is derived from the range estimate. 
	Moreover, for $\rho\in\{0.1, 0.5, 1\}$, the performance of the LS algorithm can be predicted in the high-SNR regime via the CRB analysis. Additionally, as $\rho$ increases, the minimum SNR value at which the LS algorithm aligns with the CRB trend decreases. Conversely, increasing the pilot-overhead ratio does not improve the performance of the LS approach in the low-SNR regime. Furthermore, by comparing the curves for $\rho = 0.5$ and $\rho = 1$, we can conclude that increasing the pilot-overhead ratio from  $0.5$ to $1$  results in a negligible performance improvement in both range and velocity estimations\tcb{,~whereas communication rate drops to $0$ [bits/sec]}. 
	
	\tcb{Furthermore, for $\rho = 0.1$, four different possible values of $(n_p, m_p)$ are considered and corresponding $\text{CRB}_{\text{ran}}$ and  $\text{ECRB}_{\text{vel}}$ values are presented in Table~\ref{tab:CRB_ran_vel_various_np_mp} when SNR = 5 dB. In accordance with Remark~\ref{rem:CRB},  $\text{CRB}_{\text{ran}}$ and $\text{ECRB}_{\text{vel}}$ are minimized for the largest possible values of $n_p$ and $m_p$, respectively.}

	\begin{figure}
		\centering
		\begin{tikzpicture}
		\begin{semilogyaxis}[
		width= 2.4 in,
		height=1.8 in,scale only axis,
		legend style={nodes={scale= 0.5, transform shape},at={(0,0)},anchor=south west}, 
		legend cell align={left},
		legend image post style={mark indices={}},
		xticklabel style = {font=\color{white!15!black},font=\footnotesize},
		xlabel style={font=\color{white!15!black},font=\footnotesize},
		yticklabel style = {font=\color{white!15!black},font=\footnotesize},
		ylabel style={font=\color{white!15!black},font=\footnotesize},
		ylabel={Range RMSE [m]},
		xlabel={SNR [dB]},
		xmin=-31, xmax=31,
		ymin=0.001, ymax=1000,
		xtick={-30, -20, -10 , 0, 10, 20, 30},
		ytick={0.001, 0.01,0.1,1,10, 100, 1000},
		ymajorgrids=true,
		xmajorgrids=true,
		grid style=dashed,
		every axis plot/.append style={thick},
		]

		\addplot[
		color= blue,
		line width= 1pt,
		mark = ball,
		mark options={solid},
		mark size = 1pt,
		mark indices={1, 4, 7, 10, 13, 16, 19, 22, 25, 28},
		style = solid,
		]
		coordinates {
			(-30,189.2982)(-28,187.6993)(-26,187.6994)(-24,187.5965)(-22,188.28)(-20,190.3356)(-18,187.5237)(-16,184.9361)(-14,181.5546)(-12,172.7352)(-10,163.5401)(-8,151.6341)(-6,149.5503)(-4,149.7)(-2,149.7529)(0,149.7985)(2,149.8374)(4,149.8692)(6,149.8951)(8,149.9162)(10,149.9331)(12,149.9467)(14,149.9575)(16,149.9662)(18,149.9732)(20,149.9786)(22,149.983)(24,149.9865)(26,149.9893)(28,149.9915)(30,149.9932)
		};

		\addplot[
		color = blue,
		style = solid,
		line width=1pt,
		]
		coordinates {
			(-30,31.902)(-28,25.3406)(-26,20.1288)(-24,15.9889)(-22,12.7004)(-20,10.0883)(-18,8.0134)(-16,6.3653)(-14,5.0561)(-12,4.0162)(-10,3.1902)(-8,2.5341)(-6,2.0129)(-4,1.5989)(-2,1.27)(0,1.0088)(2,0.80134)(4,0.63653)(6,0.50561)(8,0.40162)(10,0.31902)(12,0.25341)(14,0.20129)(16,0.15989)(18,0.127)(20,0.10088)(22,0.080134)(24,0.063653)(26,0.050561)(28,0.040162)(30,0.031902)
		};

		\addplot[
		color= red,
		line width= 1pt,
		mark = pentagon*,
		mark options={solid},
		mark indices={1, 4, 7, 10, 13, 16, 19, 22, 25, 28},
		mark size = 1pt,
		style = solid,
		]
		coordinates {
			(-30,259.7287)(-28,255.8776)(-26,254.8601)(-24,251.6769)(-22,244.709)(-20,230.175)(-18,205.0322)(-16,125.0038)(-14,60.8192)(-12,29.6799)(-10,1.3027)(-8,1.0264)(-6,0.81085)(-4,0.64224)(-2,0.50911)(0,0.40385)(2,0.32044)(4,0.25445)(6,0.202)(8,0.16041)(10,0.12742)(12,0.10116)(14,0.080364)(16,0.063837)(18,0.050693)(20,0.04027)(22,0.031989)(24,0.02541)(26,0.020187)(28,0.016035)(30,0.012738)
		};

		\addplot[
		color = red,
		style = dashed,
		line width= 1pt,
		]
		coordinates {
			(-30,14.1264)(-28,11.221)(-26,8.9132)(-24,7.08)(-22,5.6238)(-20,4.4672)(-18,3.5484)(-16,2.8186)(-14,2.2389)(-12,1.7784)(-10,1.4126)(-8,1.1221)(-6,0.89132)(-4,0.708)(-2,0.56238)(0,0.44672)(2,0.35484)(4,0.28186)(6,0.22389)(8,0.17784)(10,0.14126)(12,0.11221)(14,0.089132)(16,0.0708)(18,0.056238)(20,0.044672)(22,0.035484)(24,0.028186)(26,0.022389)(28,0.017784)(30,0.014126)
		};

		\addplot[
		thick,
		color= green,
		line width= 1pt,
		mark = square*,
		mark options={solid},
		mark indices={1, 4, 7, 10, 13, 16, 19, 22, 25, 28},
		mark size = 1pt,
		style = solid,
		]
		coordinates {
			(-30,271.9238)(-28,265.615)(-26,249.8841)(-24,217.0756)(-22,165.6994)(-20,67.1988)(-18,1.534)(-16,1.1998)(-14,0.94679)(-12,0.74868)(-10,0.59349)(-8,0.47035)(-6,0.37261)(-4,0.296)(-2,0.23507)(0,0.18661)(2,0.1482)(4,0.11769)(6,0.093547)(8,0.074293)(10,0.059056)(12,0.046903)(14,0.037275)(16,0.029632)(18,0.023537)(20,0.018696)(22,0.01485)(24,0.011796)(26,0.0093695)(28,0.0074399)(30,0.0059097)
		};

		\addplot[
		thick,
		color = green,
		style = dotted,
		line width= 1pt,
		]
		coordinates {
			(-30,6.3175)(-28,5.0182)(-26,3.9861)(-24,3.1663)(-22,2.5151)(-20,1.9978)(-18,1.5869)(-16,1.2605)(-14,1.0013)(-12,0.79533)(-10,0.63175)(-8,0.50182)(-6,0.39861)(-4,0.31663)(-2,0.25151)(0,0.19978)(2,0.15869)(4,0.12605)(6,0.10013)(8,0.079533)(10,0.063175)(12,0.050182)(14,0.039861)(16,0.031663)(18,0.025151)(20,0.019978)(22,0.015869)(24,0.012605)(26,0.010013)(28,0.0079533)(30,0.0063175)   
		};

		\addplot[
		thick,
		color= orange,
		line width= 1pt,
		mark = triangle*,
		mark options={solid},
		mark indices={1, 4, 7, 10, 13, 16, 19, 22, 25, 28},
		mark size = 1pt,
		style = solid,
		]
		coordinates {
			(-30,635.8311)(-28,607.0768)(-26,517.878)(-24,270.7366)(-22,118.5303)(-20,1.5109)(-18,1.1713)(-16,0.91715)(-14,0.72308)(-12,0.57187)(-10,0.4531)(-8,0.35941)(-6,0.28522)(-4,0.22697)(-2,0.18015)(0,0.14314)(2,0.1138)(4,0.090459)(6,0.071825)(8,0.057045)(10,0.045313)(12,0.036031)(14,0.028631)(16,0.022745)(18,0.018069)(20,0.014356)(22,0.011404)(24,0.0090594)(26,0.0071964)(28,0.0057174)(30,0.004542)
		};

		\addplot[
		thick,
		color = orange,
		style = dashdotdotted,
		line width= 1pt,
		]
		coordinates {
			(-30,4.4658)(-28,3.5473)(-26,2.8177)(-24,2.2382)(-22,1.7779)(-20,1.4122)(-18,1.1218)(-16,0.89104)(-14,0.70778)(-12,0.56221)(-10,0.44658)(-8,0.35473)(-6,0.28177)(-4,0.22382)(-2,0.17779)(0,0.14122)(2,0.11218)(4,0.089104)(6,0.070778)(8,0.056221)(10,0.044658)(12,0.035473)(14,0.028177)(16,0.022382)(18,0.017779)(20,0.014122)(22,0.011218)(24,0.0089104)(26,0.0070778)(28,0.0056221)(30,0.0044658)  
		};
		
		\legend{LS $(\rho = 0.02)$, $\text{CRB}_{\text{ran}}$ $(\rho = 0.02)$, LS $(\rho = 0.1)$, $\text{CRB}_{\text{ran}}$ $(\rho = 0.1)$, LS $(\rho = 0.5)$, $\text{CRB}_{\text{ran}}$ $(\rho = 0.5)$, 
			LS $(\rho = 1)$, $\text{CRB}_{\text{ran}}$ $(\rho = 1)$}	
		
		\end{semilogyaxis}
		\end{tikzpicture}
		\caption{RMSE of estimation of $\dbis$ along with $\text{CRB}_{\text{ran}}$ versus SNR [dB] when $\rho\in\{0.02, 0.1, 0.5,  1\}$.}
		\label{fig: RMSE_range_vs_SNR}
	\end{figure}

	\begin{figure}
		\centering
		\begin{tikzpicture}
		\begin{semilogyaxis}[
		width=2.4in,
		height=1.8in,scale only axis,
		legend style={nodes={scale= 0.5, transform shape},at={(0,0)},anchor=south west}, 
		legend cell align={left},
		legend image post style={mark indices={}},
		xticklabel style = {font=\color{white!15!black},font=\footnotesize},
		xlabel style={font=\color{white!15!black},font=\footnotesize},
		yticklabel style = {font=\color{white!15!black},font=\footnotesize},
		ylabel style={font=\color{white!15!black},font=\footnotesize},
		ylabel={Velocity RMSE [m/s]},
		xlabel={SNR [dB]},
		xmin=-31, xmax=31,
		ymin=0.001, ymax=1000,
		xtick={-30, -20, -10 , 0, 10, 20, 30},
		ytick={0.001, 0.01,0.1,1,10, 100, 1000},
		ymajorgrids=true,
		xmajorgrids=true,
		grid style=dashed,
		every axis plot/.append style={thick},
		]

		\addplot[
		color= blue,
		line width= 1pt,
		mark = ball,
		mark options={solid},
		mark size = 1pt,
		mark indices={1, 4, 7, 10, 13, 16, 19, 22, 25, 28},
		style = solid,
		]
		coordinates {
			(-30,63.3169)(-28,62.5803)(-26,62.5803)(-24,62.1188)(-22,62.3401)(-20,62.0811)(-18,61.6645)(-16,60.481)(-14,57.6067)(-12,47.4587)(-10,34.358)(-8,12.6135)(-6,7.75)(-4,1.9053)(-2,1.6943)(0,1.555)(2,1.4631)(4,1.4025)(6,1.3632)(8,1.3374)(10,1.3208)(12,1.3099)(14,1.3028)(16,1.2981)(18,1.2951)(20,1.293)(22,1.2917)(24,1.2907)(26,1.2901)(28,1.2896)(30,1.2893)
		};

		\addplot[
		color = blue,
		style = solid,
		line width=1pt,
		]
		coordinates {
			(-30,25.3555)(-29,22.5982)(-28,20.1406)(-27,17.9503)(-26,15.9983)(-25,14.2585)(-24,12.7079)(-23,11.3259)(-22,10.0942)(-21,8.9965)(-20,8.0181)(-19,7.1462)(-18,6.369)(-17,5.6764)(-16,5.0591)(-15,4.5089)(-14,4.0186)(-13,3.5816)(-12,3.1921)(-11,2.8449)(-10,2.5356)(-9,2.2598)(-8,2.0141)(-7,1.795)(-6,1.5998)(-5,1.4258)(-4,1.2708)(-3,1.1326)(-2,1.0094)(-1,0.89965)(0,0.80181)(1,0.71462)(2,0.6369)(3,0.56764)(4,0.50591)(5,0.45089)(6,0.40186)(7,0.35816)(8,0.31921)(9,0.28449)(10,0.25356)(11,0.22598)(12,0.20141)(13,0.1795)(14,0.15998)(15,0.14258)(16,0.12708)(17,0.11326)(18,0.10094)(19,0.089965)(20,0.080181)(21,0.071462)(22,0.06369)(23,0.056764)(24,0.050591)(25,0.045089)(26,0.040186)(27,0.035816)(28,0.031921)(29,0.028449)(30,0.025356)
			
		};

		\addplot[
		color= red,
		line width= 1pt,
		mark = pentagon*,
		mark options={solid},
		mark indices={1, 4, 7, 10, 13, 16, 19, 22, 25, 28},
		mark size = 1pt,
		style = solid,
		]
		coordinates {
			(-30,52.7727)(-28,51.9199)(-26,51.8778)(-24,50.9053)(-22,50.1909)(-20,47.5645)(-18,44.4775)(-16,26.4213)(-14,9.5207)(-12,5.9335)(-10,1.121)(-8,0.88294)(-6,0.69737)(-4,0.5523)(-2,0.43752)(0,0.34754)(2,0.27541)(4,0.21869)(6,0.17346)(8,0.13778)(10,0.10942)(12,0.086932)(14,0.06901)(16,0.054818)(18,0.043559)(20,0.034572)(22,0.027456)(24,0.021809)(26,0.017325)(28,0.013762)(30,0.010928)
		};

		\addplot[
		color = red,
		style = dashed,
		line width= 1pt,
		]
		coordinates {
			(-30,11.3393)(-29,10.1062)(-28,9.0072)(-27,8.0276)(-26,7.1546)(-25,6.3766)(-24,5.6831)(-23,5.0651)(-22,4.5143)(-21,4.0234)(-20,3.5858)(-19,3.1959)(-18,2.8483)(-17,2.5386)(-16,2.2625)(-15,2.0165)(-14,1.7972)(-13,1.6017)(-12,1.4275)(-11,1.2723)(-10,1.1339)(-9,1.0106)(-8,0.90072)(-7,0.80276)(-6,0.71546)(-5,0.63766)(-4,0.56831)(-3,0.50651)(-2,0.45143)(-1,0.40234)(0,0.35858)(1,0.31959)(2,0.28483)(3,0.25386)(4,0.22625)(5,0.20165)(6,0.17972)(7,0.16017)(8,0.14275)(9,0.12723)(10,0.11339)(11,0.10106)(12,0.090072)(13,0.080276)(14,0.071546)(15,0.063766)(16,0.056831)(17,0.050651)(18,0.045143)(19,0.040234)(20,0.035858)(21,0.031959)(22,0.028483)(23,0.025386)(24,0.022625)(25,0.020165)(26,0.017972)(27,0.016017)(28,0.014275)(29,0.012723)(30,0.011339)
		};

		\addplot[
		thick,
		color= green,
		line width= 1pt,
		mark = square*,
		mark options={solid},
		mark indices={1, 4, 7, 10, 13, 16, 19, 22, 25, 28},
		mark size = 1pt,
		style = solid,
		]
		coordinates {
			(-30,258.1163)(-28,257.1438)(-26,250.4034)(-24,224.4024)(-22,160.9944)(-20,85.2308)(-18,1.1544)(-16,0.88728)(-14,0.69148)(-12,0.54141)(-10,0.42644)(-8,0.33686)(-6,0.2667)(-4,0.21136)(-2,0.16752)(0,0.13296)(2,0.10556)(4,0.083721)(6,0.066427)(8,0.052734)(10,0.041852)(12,0.033211)(14,0.026386)(16,0.020956)(18,0.016641)(20,0.013223)(22,0.010499)(24,0.0083436)(26,0.0066271)(28,0.0052628)(30,0.0041802)
		};

		\addplot[
		thick,
		color = green,
		style = dotted,
		line width= 1pt,
		]
		coordinates {
			(-30,5.0467)(-29,4.4979)(-28,4.0087)(-27,3.5728)(-26,3.1843)(-25,2.838)(-24,2.5293)(-23,2.2543)(-22,2.0091)(-21,1.7906)(-20,1.5959)(-19,1.4224)(-18,1.2677)(-17,1.1298)(-16,1.0069)(-15,0.89744)(-14,0.79985)(-13,0.71287)(-12,0.63534)(-11,0.56625)(-10,0.50467)(-9,0.44979)(-8,0.40087)(-7,0.35728)(-6,0.31843)(-5,0.2838)(-4,0.25293)(-3,0.22543)(-2,0.20091)(-1,0.17906)(0,0.15959)(1,0.14224)(2,0.12677)(3,0.11298)(4,0.10069)(5,0.089744)(6,0.079985)(7,0.071287)(8,0.063534)(9,0.056625)(10,0.050467)(11,0.044979)(12,0.040087)(13,0.035728)(14,0.031843)(15,0.02838)(16,0.025293)(17,0.022543)(18,0.020091)(19,0.017906)(20,0.015959)(21,0.014224)(22,0.012677)(23,0.011298)(24,0.010069)(25,0.0089744)(26,0.0079985)(27,0.0071287)(28,0.0063534)(29,0.0056625)(30,0.0050467)
			
		};
		
		\addplot[
		thick,
		color= orange,
		line width= 1pt,
		mark = triangle*,
		mark options={solid},
		mark indices={1, 4, 7, 10, 13, 16, 19, 22, 25, 28},
		mark size = 1pt,
		style = solid,
		]
		coordinates {
			(-30,242.0045)(-28,228.0165)(-26,191.7514)(-24,99.5973)(-22,44.0185)(-20,1.115)(-18,0.86971)(-16,0.68348)(-14,0.53974)(-12,0.42625)(-10,0.33921)(-8,0.26867)(-6,0.21325)(-4,0.16931)(-2,0.13466)(0,0.10699)(2,0.084972)(4,0.06752)(6,0.053647)(8,0.0426)(10,0.033856)(12,0.026911)(14,0.021377)(16,0.016982)(18,0.013489)(20,0.010711)(22,0.0085084)(24,0.0067585)(26,0.0053677)(28,0.0042596)(30,0.0033836)
		};

		\addplot[
		thick,
		color = orange,
		style = dashdotdotted,
		line width= 1pt,
		]
		coordinates {
			(-30,3.5686)(-29,3.1805)(-28,2.8346)(-27,2.5263)(-26,2.2516)(-25,2.0067)(-24,1.7885)(-23,1.594)(-22,1.4207)(-21,1.2662)(-20,1.1285)(-19,1.0058)(-18,0.89638)(-17,0.7989)(-16,0.71202)(-15,0.63459)(-14,0.56558)(-13,0.50407)(-12,0.44925)(-11,0.4004)(-10,0.35686)(-9,0.31805)(-8,0.28346)(-7,0.25263)(-6,0.22516)(-5,0.20067)(-4,0.17885)(-3,0.1594)(-2,0.14207)(-1,0.12662)(0,0.11285)(1,0.10058)(2,0.089638)(3,0.07989)(4,0.071202)(5,0.063459)(6,0.056558)(7,0.050407)(8,0.044925)(9,0.04004)(10,0.035686)(11,0.031805)(12,0.028346)(13,0.025263)(14,0.022516)(15,0.020067)(16,0.017885)(17,0.01594)(18,0.014207)(19,0.012662)(20,0.011285)(21,0.010058)(22,0.0089638)(23,0.007989)(24,0.0071202)(25,0.0063459)(26,0.0056558)(27,0.0050407)(28,0.0044925)(29,0.004004)(30,0.0035686)
		};
		
		\legend{LS $(\rho = 0.02)$, $\text{ECRB}_{\text{vel}}$ $(\rho = 0.02)$, LS $(\rho = 0.1)$, $\text{ECRB}_{\text{vel}}$ $(\rho = 0.1)$, LS $(\rho = 0.5)$, $\text{ECRB}_{\text{vel}}$ $(\rho = 0.5)$, 
			LS $(\rho = 1)$, $\text{ECRB}_{\text{vel}}$ $(\rho = 1)$}	
		\end{semilogyaxis}
		\end{tikzpicture}
		\caption{RMSE of estimation of $\vbis$ along with $\text{ECRB}_{\text{vel}}$ versus SNR [dB] when $\rho\in\{0.02, 0.1, 0.5,  1\}$.}    
		\label{fig: RMSE_vel_vs_SNR}
	\end{figure}

	\begin{table}
		\begin{center}
			\begin{tabular}{| c ||c  || c |} 
				\hline
				$(n_p, m_p)$  & $\text{CRB}_{\text{ran}}$ &  $\text{ECRB}_{\text{vel}}$ \\
				\hline
				$(1, 11)$  &   0.2511 & 0.1862\\
				\hline
				$(2, 5)$  &    0.2512 & 0.2016\\
				\hline
				$(5, 2)$  &   0.2517  & 0.2008 \\
				\hline
				$(11, 1)$  &     0.2306 &  0.2007 \\
				\hline
			\end{tabular}
		\end{center}
		\caption{\scriptsize{\tcb{$\text{CRB}_{\text{ran}}$ and  $\text{ECRB}_{\text{vel}}$ versus $(n_p, m_p)$ when $\rho = 0.1$ and SNR = 5 dB .}}}
		\label{tab:CRB_ran_vel_various_np_mp}
		\vspace{-0.25cm}
	\end{table}

	\section{Concluding Remarks}
	In this work, closed form expressions for the theoretical performance bounds of OFDM based bi-static ISAC system are derived as a function of the pilot pattern. It is numerically verified that when the pilot-overhead ratio satisfies conditions for reliable range and velocity estimation, LS algorithm performs very close to the theoretical performance bounds in the high-SNR regime. In other words, by using the closed form expressions \eqref{eq:CRB_ran} and \eqref{eq:CRB_vel}, performance of the LS-algoritm can be predicted, thus the optimal pilot patterns can be designed according to \eqref{eq:CRB_ran} and \eqref{eq:CRB_vel} to minimize range or velocity estimation errors. \tcb{Instead of periodic pilot patterns, non-periodic pilot patterns can be also considered to improve the sensing performance.} One possible extension of this work is incorporating a LoS path between the transmitter and receiver. Additionally, extending the analysis to a multi-target scenario and deriving the corresponding CRB expressions are other promising future directions. 
	
			\appendices
			\section{Fisher Information Matrix Calculations}\label{sec:FisherInformationMat}
			The following derivatives can be easily computed:
			\begin{align*}
			\frac{\partial \mu_{n,m}(\etab)}{\partial \alpha_R} &= e^{j 2\pi \left(f_D m \Ts-\tau n \Delta f\right)} X_{n,m},\\
			\frac{\partial \mu_{n,m}(\etab)}{\partial \alpha_I} &= je^{j 2\pi \left(f_D m \Ts-\tau n \Delta f\right)} X_{n,m},\\
			\frac{\partial \mu_{n,m}(\etab)}{\partial f_D} &= j \alpha 2\pi m \Ts e^{j 2\pi \left(f_D m \Ts-\tau n \Delta f\right)}X_{n,m}, \\
			\frac{\partial \mu_{n,m}(\etab)}{\partial \tau} &= -j \alpha 2\pi n \Delta f e^{j 2\pi \left(f_D m \Ts-\tau n \Delta f\right)}X_{n,m}.
			\end{align*}
			For example,  $[\Jb(\etab)]_{1,1}$ can be computed as follows:
			\begin{align}
			&[\Jb(\etab)]_{1,1} \nonumber\\
			&= \sum_{(n,m)\in\mathcal{P}} \abs{e^{-j 2\pi \left(f_D m \Ts-\tau n \Delta f\right)}}^2\abs{X_{n,m}}^2 = \aP.
			\end{align}
			Similarly, other entries of $\Jb(\etab)$ can be computed as follows:
			\begin{align}
			\Jb(\etab) = \frac{2}{\sigma^2} 
			\begin{bmatrix}
			\Ab   & \Bb \\
			\Cb   & \Db  \\
			\end{bmatrix}
			\end{align}
			where
			\begin{align}\label{eq:A}
			\Ab \triangleq  \aP
			\begin{bmatrix}
			1   & 0 \\
			0  & 1 
			\end{bmatrix},
			\end{align}
			\begin{align}\label{eq:B}
			\Bb \triangleq  2\pi
			\begin{bmatrix}
			-\alpha_I  \Ts \sum_{(n,m)\in \mathcal{P}} m   & \alpha_I  \Delta f \sum_{(n,m)\in \mathcal{P}} n \\
			\alpha_R  \Ts \sum_{(n,m)\in \mathcal{P}} m   & -\alpha_R  \Delta f \sum_{(n,m)\in \mathcal{P}} n 
			\end{bmatrix},
			\end{align}
			\begin{align}\label{eq:C}
			\Cb \triangleq  2\pi
			\begin{bmatrix}
			-\alpha_I  \Ts \sum_{(n,m)\in \mathcal{P}} m   &  \alpha_R  \Ts \sum_{(n,m)\in \mathcal{P}} m\\
			\alpha_I  \Delta f \sum_{(n,m)\in \mathcal{P}} n   & -\alpha_R  \Delta f \sum_{(n,m)\in \mathcal{P}} n 
			\end{bmatrix},
			\end{align}
			and
			\begin{align}\label{eq:D}
			\Db \triangleq  4\pi^2 \abs{\alpha}^2
			\begin{bmatrix}
			\Ts^2 \sum_{(n,m)\in \mathcal{P}} m^2   &  -\Ts \Delta f \sum_{(n,m)\in \mathcal{P}} nm\\
			- \Ts \Delta f \sum_{(n,m)\in \mathcal{P}} nm  &  \Delta f^2 \sum_{(n,m)\in \mathcal{P}} n^2
			\end{bmatrix}.
			\end{align}

			\section{Fisher Information Matrix Terms Calculation for Periodic Pilot Patterns}\label{sec:FIM_pilot_periodic}
			One can easily compute the following summations:
			\begin{align}
			\sum_{(n,m)\in\mathcal{P}} n &= (L+1) \sum_{k = 0}^{K} k n_p = \frac{K \aP n_p}{2}, \label{eq:PN} \\
			\sum_{(n,m)\in \mathcal{P}} m &= (K+1) \sum_{\ell = 0}^{L} \ell m_p = \frac{L \aP m_p}{2},\label{eq:PM} \\
			\sum_{(n,m)\in\mathcal{P}} nm & = \sum_{k = 0}^{K}\sum_{\ell = 0}^{L} k \ell n_p m_p = \frac{KL\aP m_pn_p}{4},\label{eq:PNM}\\
			\sum_{(n,m)\in\mathcal{P}} n^2 &= (L+1) \sum_{k = 0}^{K} k^2 n_p^2 = \frac{K(2K+1) \aP n_p^2}{6},\label{eq:PN2} \\
			\sum_{(n,m)\in\mathcal{P}} m^2 &= (K+1) \sum_{\ell = 0}^{L} \ell^2 m_p^2 = \frac{L(2L+1) \aP m_p^2}{6}. \label{eq:PM2}
			\end{align}
			By using \eqref{eq:PN}, \eqref{eq:PM} and \eqref{eq:PNM}, $Q_{NM}$ can be expressed as
			\begin{align}\label{eq:Q_NM_periodic}
			Q_{NM} = \frac{KL\aP m_pn_p}{4}- \frac{K \aP n_p}{2} \frac{L \aP m_p}{2} \frac{1}{\aP} = 0.
			\end{align}
			Similarly, by using \eqref{eq:PN}, \eqref{eq:PM}, \eqref{eq:PN2} and \eqref{eq:PM2}, we can write
			\begin{align}\label{eq:Q_N2_periodic}
			Q_{N^2}  &=  \frac{K(2K+1) \aP n_p^2}{6} - \frac{K^2 \aP n_p^2}{4}  = \frac{K (K+2)\aP n_p^2}{12}
			\end{align}
			and
			\begin{align}\label{eq:Q_M2_periodic}
			Q_{M^2}  &=  \frac{L(2L+1) \aP m_p^2}{6} - \frac{L^2 \aP m_p^2}{4}  = \frac{L (L+2)\aP m_p^2}{12}.
			\end{align}

	\bibliographystyle{IEEEtran}
	\bibliography{bibfile}
	
\end{document}